\documentclass[10pt, a4paper]{amsart}

\usepackage{amsmath}
\usepackage{latexsym}
\usepackage[all]{xy}
\usepackage{color}
\newtheorem{teorema}{Theorem}[section]

\include{amslatex}
\newtheorem{proposicion}[teorema]{Proposition}

\usepackage[margin=1.0 in]{geometry}

\newtheorem{comentario}[teorema]{Remark}

\numberwithin{equation}{section}

\begin{document}
\begin{title}[Emergent classical gravity from emergent quantum mechanics]
 {On the emergence of gravity in a framework of emergent quantum mechanics}
\end{title}
\date{\today}
\maketitle
\begin{center}
\author{Ricardo Gallego Torrom\'e\footnote{email: rigato39@gmail.com}}
\end{center}
\begin{center}
\address{Department of Mathematics\\
Faculty of Mathematics, Natural Sciences and Information Technologies\\
University of Primorska, Koper, Slovenia}
\end{center}

\begin{abstract}
In this work we consider the emergent origin of gravity in the framework of Hamilton-Randers theory, a theoretical framework for emergent quantum mechanics. After presenting the essential ingredients of the theory, a derivation of the weak equivalence principle from the first principles of the theory follows. Then it is shown that the Newtonian model of gravity and certain modified gravity models are consistent with Hamilton-Randers theory. One of these models has strong formal resemblance with MOND, providing further ground to think that Hamilton-Randers theory could be a natural theoretical frame for realistic modified gravity models.
\end{abstract}


\section{Introduction}

In relativistic theories, spacetime is described as a dynamical, geometric object characterized by being affected and affecting the dynamics of matter and fields living over it. This is in sharp contrast with the usual description of physical phenomena offered by quantum mechanics, where the back-ground spacetime structure is fixed. Being general relativity and quantum mechanics theories aimed to be of universal validity and due to the differences in the way they describe physical processes, one should expect to find predictions associated to general relativity in direct conflict with predictions of quantum mechanical models.

Such conceptual and technical confrontation between quantum mechanics and general relativity must be superseded by a new conceptually and mathematically consistent unified theory. It is usually stated that such a theory must be a quantum mechanical theory of the gravitational interaction. There are several research programs aimed to solve the incompatibility problems between general relativity and quantum theory. String theory and loop quantum gravity are notable frameworks. Although the expectation on these programs are legitime and well grounded, the problem of finding a consistent theory merging quantum mechanics and general relativity or suitable modifications of them is still elusive and remains open. Therefore, different avenues for the analysis of these issues can be of certain interest.

A new perspective on the foundational issues of the quantum theory are suggested under the common name of {\it emergent quantum mechanics}.
These programs share the point of view that there is an underlying more radical level of physical description from where quantum mechanics is obtained as an effective description \cite{Adler, Blasone2, Elze, Ricardo2005,Ricardo2006,Ricardo2014, Groessing2013, Hooft, Hooft 2001, Hooft2006,Hooft2012wavefunctionschroedingercollapse, Hooft2016, Smolin2012}.

For several frameworks of emergent quantum mechanics the degrees of freedom at the fundamental scale are deterministic and local \cite{Hooft2016,Ricardo2014}. One fundamental difficulty in such approaches is that the associated Hamiltonian operators, being linear in the momentum operators, are not bounded from below.
Therefore, in order to ensure the existence of stable minimal energy states, a natural requirement for the construction of viable quantum models of matter, a mechanism to stabilize the vacuum is necessary. One of the mechanism proposed in the literature involves a dissipative dynamics at the fundamental Planck scale \cite{Blasone2, Elze, Hooft}. It was suggested that the gravitational interaction plays an essential role as the origin of the information loss dynamics and must be present at the level of the fundamental  scale. However, gravitational interaction could be also a classical and emergent phenomenon, absent at the fundamental scale  where it is assumed that the dynamics of the microscopic fundamental degrees of freedom takes place. If this is the case, it is not natural to appeal from the beginning to gravity as the origin of the dissipation of information at the fundamental level of physical description.

The research line on emergent quantum mechanics developed by the author turns around the conjecture that the quantum description of physical systems are emergent from a deeper level description of Nature. At such deeper level, the fundamental dynamics is compatible with locality in the configuration space of Hamilton-Randers dynamical systems and it is also compatible with determinism \cite{Ricardo2005,Ricardo2006,Ricardo2014}, although such a dynamics is highly non-trivial in the sense of complexity. Emergent quantum mechanics as envisaged in the formulation of the author, that due to the dynamical structure it was named {\it Hamilton-Randers theory}, postulates a set of formal properties for the mathematical objects and dynamics describing the fundamental degrees of freedom evolution. The theory develops both, a geometric formulation of the dynamics that help to keep track of general covariance and the Koopman-von Neumann formulation of the dynamics \cite{Koopman1931,von Neumann}, that is instrumental to show how the quantum mechanical description of the theory emerges as a coarse grained average description of the fundamental dynamics \cite{Ricardo2014}.

One of the consequences of the theory is the realization of the emergent origin of gravity. The argument is based upon the formal derivation of the weak equivalence principle from the assumptions of emergent quantum mechanics \cite{Ricardo2014,Ricardo2019a} and the consistency of the fundamental interaction in certain domain with the fundamental properties partially characterizing gravitational models as the existence of a causal structure, diffeomorphism invariance and general covariance \cite{Ricardo2014}. In particular, the assumptions of Hamilton-Randers theory directly involved in the derivation of the weak equivalence principle concern the applicability of the mathematical theory of concentration of measure in mm-Gromov spaces \cite{Gromov,MilmanSchechtman2001,Berger 2002}. Concentration of measure that we use as a guiding principle is the property that in certain large dimensional spaces endowed with a measure and topological metric, $1$-Lipschitz functions or operators are almost constant almost everywhere. Concentration of measure appears in multiple faces of mathematics, as in metric geometry \cite{Gromov,Berger 2002}, functional analysis \cite{MilmanSchechtman2001} or in probability theory \cite{Talagrand}. In mathematical physics, concentration of measure is pivotal in the geometric-like derivation of Maxwell-Boltzmann distribution \cite{Gromov}, for instance. The implications of concentration of measure to emergent quantum mechanics rely on its fundamental role in a natural spontaneous quantum state collapse mechanism that appears in the theory and, in a very closely related way, for the interpretation that the fundamental interaction during the natural spontaneous collapse processes is the gravitational interaction (see \cite{Ricardo2014} chapters 6 and 7 and also ref. \cite{Ricardo2019a}).

In order to illustrate these concepts we provide in this paper a short introduction to Hamilton-Randers theory explaining the appearance of the $1$-Lipschitz dynamics condition and the use of concentration of measure. The introduction is rather compact, but the interesting reader can find details and further developments in \cite{Ricardo2014} and in the references to articles of the author there in. After such introduction we describe our derivation of the weak equivalence principle from the principles of the theory.  This new expression of the concentration of measure relies on the existence of a domain where the Hamiltonian interaction and hence, the evolution operators, is $1$-Lipschitz (see ref. \cite{Ricardo2014}, {\it section} 7.1 for a discussion of this statement). In addition to other formal properties of the Hamilton-Randers dynamical models as the existence of a causal structure, diffeomorphism invariance and general covariance, the formal derivation of the weak equivalence principle from the assumptions of Hamilton-Randers theory provides support to identify the fundamental interaction of emergent quantum mechanics in the domain ${\bf D}_0$ with the gravitational interaction.

In this work we reconsider the derivation of the weak equivalence principle from the principles of Hamilton-Randers theory. This revised presentation is an improvement of a previous argument \cite{Ricardo2019a}. In particular, it is remarked that the weak equivalence principle applies irrespectively from the sources being classical or quantum. This leads to a classical, possibly fluctuating, and emergent character for gravitation. After this, it is shown that several models of non-relativistic gravitational models are compatible with Hamilton-Randers theory. This includes the Newtonian model of gravity, but also certain modified Newtonian gravity based upon a logarithmic potential model as well as convex combinations of Newtonian and the logarithmic potential models. Such models are shown to be consistent with the $1$-Lipschitz condition. This finding suggests that Hamilton-Randers theory could be a theoretical framework for certain modified Newtonian gravity models beyond general relativity.
Different from MOND \cite{Milgrom 1983a,Milgrom 1983b, Milgrom 2015} and AQUAL \cite{Bekenstein Milgrom} models, that contain an universal scale of acceleration $a_{MOND}$, the first model discussed of modified gravity contains an universal characteristic length that, as we will show, it is very different from MOND characteristic length. Unfortunately, this first model is not consistent with the barionic Tully-Fisher relation.
A second model of modified gravity is also discussed. As in MOND, it contains a critical acceleration which is universal and there is no universal length scale. This second model is also consistent with the barionic Tully-Fisher relation. However, the model explicitly violates the third law of Newtonian dynamics, also a feature from Milgrom's theory, but being consistent with $1$-Lipschitz condition of Hamilton-Randers theory, it is also consistent with the weak equivalence principle. Indeed, the theory has several key points in common with MOND formulation, albeit the assumptions and development are different as in MOND.

These examples illustrate how consistency with the Hamilton-Randers theory can help to shape viable models of gravitation from the point of view of emergent quantum mechanics, by imposing consistency with Hamilton-Randers theory as a necessary condition for modified gravity models.

\section{Hamilton-Randers theory in a nut-shell}
Emergent quantum mechanics theories are certain theoretical frameworks that aim to reproduce quantum mechanics or modifications of quantum mechanics as an effective description of a different dynamics and ontology beneath the usual quantum dynamics or usual ontology. There are several different approaches to emergent quantum mechanics \cite{Adler,Hooft2016,Smolin2012,Singh,Ricardo2014}. We consider in this study the approach developed by the author \cite{Ricardo2014}.

 Hamilton-Randers theory is based upon the following assumptions:
\begin{enumerate}
\item There is a local, deterministic dynamics beneath quantum dynamical systems. Locality is not understood in the standard spacetime framework sense, but in the context of high dimensional configuration space for the fundamental dynamical systems. Our dynamical systems are highly non-local from the point of view of spacetime framework.

\item The number of the fundamental degrees of freedom associated to a quantum system is in relation $N$ to $1$, with $N\gg 1$. $N$ is not universal, but depends on the system. For instance, a dynamical system describing a single electron or a single neutrino contains $N\gg 1$ degrees of freedom and the ration $N:1$ is a {\it measure} of the complexity of the system. Indeed, inertial mass appears as a measure of the complexity of the system too \cite{Ricardo2014,Ricardo2023}.

\item  The dynamical models describe the dynamics of the fundamental system in an abstract configuration tangent manifold $TM$. The dynamics of the fundamental degrees is such that there is an uniform or universal upper bounded speed and proper acceleration. The dynamics is formulated using geometric models by means of a Hamiltonian constructed from a time-symmetrization of Randers-type Hamiltonian \cite{Randers}. This kind of Hamiltonian are linear perturbation of the kinematical Hamiltonian. The perturbation has its grounds in an irreversible character of the fundamental evolution. In order to have a geometric interpretation of the dynamics in terms of well-defined connections, a limit on the speed and acceleration of the fundamental degrees of freedom (see \cite{Ricardo2014}, Chapter 3).

\item An alternative formulation of the fundamental dynamics makes use of Koopman-Von-Neumann theory (see \cite{Ricardo2014}, Chapter 4). This step is key to relate the fundamental dynamics with an coarse grained description of quantum mechanics.
    \begin{comentario}
    The use of Koopmann-von Neumann finds a parallel development in the work of G. 't Hooft cellular automaton interpretation of quantum mechanics \cite{Hooft2016}, but the identification of the fundamental degrees of freedom is different than in Hooft's approach. Our dynamical models have also certain similarities with the models that one found in the theory of {\it trace dynamics} developed by S. L. Adler \cite{Adler}, where sub-quantum degrees of freedom are described in the form of unitary matrices.
\end{comentario}

\item The dynamics $U_t$ of the fundamental degrees of freedom is assumed to have three different phases which are approximately cyclically repeated, defining the {\it fundamental cycles}: an expansive-ergodic phase followed by a contractive phase followed by an expansive phase.
    \begin{comentario}
    The relevant phase for the probabilistic, non-local description of quantum mechanics is the ergodic phase of the fundamental cycles. By applying ergodic concepts, a fiber averaging is interpreted as a $t$-time averaging. Since this happens in the ergodic regime,  the averaging description leads to the interpretation of the wave function as probability amplitude, and in general, to the interpretation of the quantum state as a coarse grained mathematical object associated to the deeper description.
    \end{comentario}

\item The Hilbert space of the quantum mechanical descriptions emerges by means of a {\it coarse grained averaging operation} (see \cite{Ricardo2014}, {\it Chapter 5}) from the category of Hamilton-Randers dynamical systems to the category of Hilbert spaces indexed over the $4$-dimensional spacetimes. The averaging operation is applied to several mathematical objects that appear in the Koopman-von Neumann formulation of the Hamilton-Randers dynamical systems, including the image quantum states and the inner product of the Hilbert space used for the quantum representation of the system as well as to the dynamics. It turns out that the hypothesis of emergence is equivalent to the functorial character of the averaging operation.

    \item In the theory there are two classes of time parameters, $(t,\tau)\in\,(I_1,I_2)\subset\,(\mathbb{K}_1,\mathbb{K}_2)$, where $\mathbb{K}_1,\mathbb{K}_2$ are number fields. For most of the developments of the theory both number fields $\mathbb{K}_1,\mathbb{K}_2$ are taken as isomorphic to the real numbers field $\mathbb{R}$.
     \begin{comentario}
      The two time parameters play a similar role to the slow and fast time parameters in classical dynamics. However, the time parameters $t$ and $\tau$ are not in one-to-one relation as it is usual the case in the theory of perturbations in classical dynamics \cite{Arnold}. This two-dimensional character of time is an essential element of Hamilton-Randers theory. It is of relevance for the ergodic interpretation of quantum mechanical objects as averages of functions and operators defined at the sub-quantum and in the description of non-local quantum correlations and interpretation of entanglement \cite{Ricardo2017b,Ricardo2014}.
Note also that our notion of $2$-dimensional time is also different than the one discussed in two-dimensional physics extensively by I. Bars \cite{Bars2001}. In our case, the $t$-time parameters used in macroscopic description (this includes in our language also standard quantum mechanical systems) have an emergent origin associated with counting fundamental cycles and they do not have associated an associate geometric dimension, as it is the case of Bar's $2$-time physics theory.
     \end{comentario}

    \begin{comentario}
    The cyclic nature of the fundamental interaction appeared first in the preliminary versions of our theory \cite{Ricardo2005,Ricardo2006} and also, in the interpretation of quantum mechanics advocated by D. Dolce in \cite{Dolce} and subsequent works. Despite this similarity, Dolce's theory and the theory advocated by the author are rather different. In Dolce's theory the notion of time is the usual one as it appears in current physical models, but where physical fields are subjected to cyclic boundary conditions with the periods cycle given by the associated inverse Compton frequency, while in Hamilton-Randers theory the fundamental almost-cyclic evolution is parameterized by a different time parameters than usually and the relation between the period and the mass of the system is exponential \cite{Ricardo2014,Ricardo2023}.
\end{comentario}
\end{enumerate}

The mechanism in Hamilton-Randers theory to explain the strange spookiness of quantum non-locality can be described as follows \cite{Ricardo2014,Ricardo2017b}.
Consider two spacelike separated points $x_1$ and $x_2$ of the spacetime $M_4$ and let us consider two quantum systems $\psi_1$ and $\psi_2$ such that they are associated to the points $x_1$ and $x_2$ respectively. During the evolution of the fundamental dynamics, which is parameterized by a $t$-time parameter, the sub-quantum degrees of freedom associated to $\psi_1$ and $\psi_2$ interact during the fundamental evolution, that took place in a large dimensional configuration space $TM$, not directly on the spacetime. Such interactions happen even if the points $x_1$ and $x_2$ are spacelike separated. This is because the sub-quantum degrees of freedom perform an ergodic like evolution during each fundamental cycle of the fundamental dynamics. Therefore, the evolution is filling the available phase space $T^*TM$ in a finite time in a sufficient dense way. But according to Hamilton-Randers theory, after the ergodic regime there is a natural, instantaneous contractive phase as a consequence of the geometric concentration of measure. Once the contractive dynamics takes place, the system evolves towards a reduced space. This is the reason why one can assign in our case $\psi_1$ to $x_1$ and $\psi_2$ to $x_2$. This attribution is consistent with quantum conserved quantities like momentum conservation (see Chapter 5 in \cite{Ricardo2014}. Furthermore, the filling must be consistent with conserved quantities. Thus for some systems it will appear entanglement related with the initial correlations of the systems due to interactions between the sub-quantum degrees of freedom, while for other systems they will remain independent from each other and they will not be correlations.

When the systems naturally collapse as a consequence of concentration of measure, physical attributes can be associated to the system and a measurement of them can be performed, bringing its a value that can be associated as the value of the observable for the given system.  This process happens during each individual cycle.
However, when the $t$-time dependence and dynamics is dismissed from the description of the system and only the slow $\tau$-time dynamics is considered, then the process of interaction during the ergodic regime brings the instantaneous and non-local characteristics of the  quantum description.

Being the mechanism described above highly non-local from the point of view of spacetime, it can potentially avoid the assumptions of Bell's theorem and also the assumptions underlying Kochen-Specker theorem. Such results are typical no-go theorems. As such, they based upon hypothesis and assumptions. For the cases under consideration, the assumptions of relevance are locality for the hidden variable models in the case of Bell's theorem and statistical independence in the case of Bell's theorem; contextuality avoids the Kochen-Specker argument (see for instance chapter 9 in \cite{Isham 1995} and \cite{Redhead}). In the case of Hamilton-Randers models, the non-locality is of such nature that it avoids the assumptions leading to Bell inequalities and its generalizations (see \cite{Ricardo2014}, chapter 8). The conjecture is that the degree of non-locality and because it is due to direct interactions at the level of sub-quantum degrees of freedom implies contextuality. If this conjecture is true, then the Kochen-Specker theorem is not applicable directly.
 Of course, it could still happen that similar results to Bell's theorems or Kochen-Specker theorem still apply to Hamilton-Randers theory. Therefore, further analysis are required to complete the understanding of the theory and to have a completely developed theory.

 The relevant phase of the fundamental dynamics for the topic considered in this paper is the {\it contractive phase}. Sufficient conditions for these contractive phases to happen is the occurrence of a breaking of ergodicity and that the evolution operator $U_t$  is dominated by a $1$-Lipschitz operator. Expressed quite pictorially, the $1$-Lipschitz conditions means that the intensity of the dynamics is bounded and that {\it geometry dominates}. Indeed, it can be shown that for the Hamiltonian functions considered in Hamilton-Randers theory, it is possible to decompose the Hamiltonian (an hence, the dynamics) in a $1$-Lipschitz component and in a non-Lipschitz component \cite{Ricardo2014} in certain domain of the configuration space. In the $1$-Lipschitz domain we are in position to apply the mathematical theory of concentration of measure \cite{Gromov,MilmanSchechtman2001,Talagrand}.  Concentration of measure  means that for any regular $1$-Lipschitz function in a large dimensional $mm$-Gromov space, the values that the function can take are almost constant almost everywhere, with bounds in the dispersion that typically depend exponentially with the dimension of the space. Such inequalities are geometric generalizations of the Chernoff's bounds that appear in probability theory \cite{Talagrand}. In the context of emergent quantum mechanics, concentration of measure leads from the assumption on the existence of the contractive phase to the situation where there is a natural spontaneous collapse of the wave function. This is the domain of the fundamental interaction between a contractive domain and the next expansive and ergodic domains of each fundamental cycle. This contractive phase, where the dynamics is close to a Hamiltonian identically zero is the {\it the classicality domain}, where any potentially measurable observable of the system by a physical observer has well defined values. Note that this only happens in that specific domain and that in the rest of the non-classical domain of the fundamental interaction, macroscopic observable do not need to reach well-defined values.

\section{On the weak equivalence principle in Emergent Quantum mechanics}
In this section we review our previous argument on the emergent character of the weak equivalence principle in Hamilton-Randers theory. We follow the exposition of the argument that can be found in \cite{Ricardo2014} for more details.
Let us consider a general dynamical system $\mathcal{S}$ as composed by two sub-systems $A$ and $B$. The configuration space is a tangent manifold $TM$, with $M$ being a product manifold $M=\prod^N_{k=1}\,M^k_4$ of $N\gg 1$ diffeomorphic four dimensional manifolds $\varphi_k:M^k_4\to M_4$ to a $4$-dimensional spacetime manifold model $M_4$. This construction is motivated by the assumption that the fundamental degrees of freedom are identical and  the assumption that the spacetime description of physical events should be obtained as an effective description of the underlying theory.
We consider local smooth functions defined on open sets $\mathcal{O}\subset \,T^*TM$,
\begin{align*}
 X^\mu:\mathcal{O}\to \mathcal{U}\subset\mathbb{R},\quad (u^{k_1},...,u^{k_N},p^{k_1},...,p^{k_N})\mapsto X^\mu(u^{k_1},...,u^{k_N},p^{k_1},...,p^{k_N}),\quad \mu=1,2,3,4,
\end{align*}
 the {\it macroscopic coordinates} associated to the physical system $\mathcal{S}_i$.

According to the assumptions of Hamilton-Randers theory, there is a domain of the time evolution where the dynamical evolution is such that the whole Hamiltonian is zero. This happens at the metastable points $\{t\to (2n+1)T,\,n\in\,\mathbb{Z}\}$. The metastable domain is denoted by ${\bf D}_0$. It is an open set containing the metastable points and includes in the contractive phase of the evolution. Since the fundamental dynamics is almost-cyclic, the metastable domain contain discrete domains, labeled by integers associated to different fundamental cycles.
Under the conditions of universal bounded acceleration and speed for the fundamental degrees of freedom and because the Hamiltonian is close to zero, one can show that the functions $X^\mu((2n+1)T)=X^\mu(\tau)$  are $1$-Lipschitz in $t$-time parameter on each cycle parameterized by $\tau\equiv n\in\,\mathbb{Z}$.

This construction is an example of emergent objects, in this case, emergent coordinate systems associated $\{X^\mu(\tau)\}$ along a curve parameterized by the emergent parameter $\tau$. The parameter $\tau$ takes integer values. However, if we consider that the precision is very high compared with usual intervals measured by macroscopic observers, we can make the approximation that $\tau$ is a real parameter, $\tau \in\,I\subset \mathbb{R}$.

Let us consider two subsystems $A$ and $B$ of the full system under consideration $\mathcal{S}$. The sub-systems $A$, $B$ are embedded in $\mathcal{S}$ such that
\begin{align}
\mathcal{S}=\,A\sqcup B,
 \end{align}
 for a well defined union operation $\sqcup$ for systems composed by fundamental degrees of freedom. Let us also consider local coordinate systems such that the identification
   \begin{align}
   A\equiv (u_1(\tau),...,u_{N_A}(\tau),0,...,0) \quad \textrm{and}\quad B\equiv (0,...,0,v_1(\tau),...,v_{ N_B}(\tau)),
   \label{embedding A,BtoS}
   \end{align}
   with $N=\,N_A+N_B,N_A,N_B\gg 1$ holds good.
The whole system $\mathcal{S}$ has associated represented in local coordinates in $TM$ of the form
   \begin{align*}
   \mathcal{S}\equiv (u_1(\tau),...,u_{ N_A}(\tau),v_1(\tau),...,v_{ N_B}(\tau)).
   \end{align*}
    By the action of diffeomorphisms $\varphi_k:M^k_4\to M_4$, one can consider the world lines of the fundamental degrees of freedom on $M_4$ at each constant value of $t$ modulo $2T$.  Therefore, it is reasonable to define the {\it macroscopic coordinates of the system} to be
    \begin{align}
    \tilde{X}^\mu_i(\tau(n))=\frac{1}{N}\,\lim_{t\to (2n+1)T}\sum^{N_i}_{k_i=1}\,\varphi^\mu_{k_i}(x_{k_i}(t)),\quad\,i=A,B,\mathcal{S},\,\mu=0,1,2,3,
    \label{definicionofXmu}
    \end{align}
    where here $\varphi^\mu_{k_i}$ are local coordinates on $M_4$, defined after the action of the diffeomorphism $\varphi_{k_i}$.
  We identify  $\tau(n)$ with the $\tau$-time parameter and consider it continuous in relation with macroscopic or quantum time scales. Since $\tau(n)$ is not trivially the identity, this construction is more general than above, but the significance is analogous. Then by \eqref{embedding A,BtoS} we have that
      \begin{align}
    \tilde{{X}}^\mu(\tau)=\frac{1}{N}\,\lim_{t\to (2n+1)T}\sum^{N}_{k=1}\,\varphi^\mu_{k_i}(x_{k_i}(t)),\quad\,i=A,B,\mathcal{S},
    \label{definicionofXmu2}
    \end{align}
    where here $n$ is equivalent to $\tau$ as indicator of the fundamental cycle. Now we observe that the construction of $ \tilde{{X}}^\mu(\tau):\mathbb{R}\to \mathbb{R}$ suggest to consider a coordinate patch in $M_4$ such that all these curves are represented locally in such coordinate patch by the values $ \tilde{{X}}^\mu(\tau)$. Such coordinate system $U, \tilde{{X}}^\mu(\tau))$ is assumed to be in the atlas of $M_4$.

Given the geometric structure of the theory, there is a standard measure on $T^*TM$ that we denote by $\mu_P$. This is the measure coming from the pull-back of the standard measure in $\mathbb{R}^{8N}$. With this measure, we can define the median of an arbitrary function $f:T^*TM\to \mathbb{R}$, that is the numerical value $M_f$ with the property
\begin{align*}
\mu_P (f>M_f)=1/2=\mu_P(f<M_f).
\end{align*}
 In particular, we can consider the median of the functions $ \tilde{X}^\mu_i(\tau)$, $M^\mu(\tau).$
Note that the median coordinate functions $M^\mu(\tau)$ only depend on the preparatory macroscopic conditions and on the particular cycle where it is being considered, the cycle marked by $\tau$.

\subsection{Proof of the Weak Equivalence Principle}
A function $f:{\bf T}_1\to {\bf T}_2$ between two metric spaces $({\bf T}_1,d_1)$ and $({\bf T}_2,d_2)$ is $1$-Lipschitz if
\begin{align*}
d_2(f(x),f(y))< C \,d_1(x,y),\quad 0<C<1.
\end{align*}
We can provide an heuristic interpretation of the concentration of $1$-Lipschitz functions.
Let $f:{\bf T}\to \mathbb{R}$ be a $1$-Lipschitz function on a normed topological space $({\bf T},\|,\|_{\bf T})$ locally homeomorphic to  $\mathbb{R}^N$. Then the $1$-Lipschitz condition is a form of {\it equipartition} of the  variation of $f$ originated by a variation on the point on the topological space ${\bf T}$ where the function $f$ is evaluated. When the dimension of the space ${\bf T}$ is very large compared with $1$, the significance of the $1$-Lipschitz condition is that $f$ cannot admit large {\it standard variations} caused by the corresponding standard variations on the evaluation point of {\bf T}. Otherwise, a violation of the $1$-Lipschitz condition can happen, since the large dimension provides long contributions to the variation of $f$ by adding each of the contributions coming for each individual variable on which $f$ depends on.

 For the case that we are considering, with real functions on the configuration space considered in emergent quantum mechanics, concentration of measure  \cite{Talagrand} implies the following Chernoff's type bound,
\begin{align}
\mu_P\left(\left|f-M_f\right|\,\frac{1}{\sigma_f}\,>\frac{\rho}{\rho_P}\right)\leq \, \frac{1}{2} \exp\left(-\frac{\rho^2}{2\rho^2_P}\right),
\label{concentration2}
\end{align}
 In particular, for the $U_t$ dynamics in the Lipschitz dynamical regime ${\bf D}_0$, the $\tau$-evolution of the coordinates $\tilde{X}^\mu(\mathcal{S}(\tau))$, $\tilde{X}^\mu(A(\tau))$ and $\tilde{X}^\mu(B(\tau))$ with same initial conditions differ between each other after the dynamics at $\tau$-time in such a way that the condition
   \begin{align}
  \mu_P\left(\frac{1}{\sigma_{\tilde{X}^\mu}}\,|\tilde{X}^\mu(\mathcal{S}_i(\tau))-M^\mu(\tilde{X}^\mu_i))|>\rho\right)_{t\to (2n+1)T}\sim C_1\exp \left(-\,C_2 \frac{\rho^2}{2\,\rho^2_p}\right),
\label{generalconcentrationofmeasure}
\end{align}
$\mu=1,2,3,4,\,i=A,B,\mathcal{S}$ must hold almost everywhere in the configuration space.
The constants $C_1,C_2$ are of order $1$, where $C_2$ depends on the dimension of the spacetime $M_4$. $\rho_p$ is  independent of the system $i=A,B,\mathcal{S}$. $\sigma_{\tilde{X}^\mu}$ is associated to the highest resolution possible in the measurement of ${\tilde{X}^\mu}$.
Moreover, since the $N$ identical degrees of freedom are identical, it is reasonable to assume that
  \begin{align}
\frac{\rho}{\rho_P}\sim\, N.
\label{scales}
\end{align}
That is, the contribution to the dispersion from each degree of freedom are of the same order.

Sources and test particles are different.
A test particle system is describe by a dynamical system such that the $\tau$-evolution of the center of mass coordinates $\tilde{X}^\mu$  are determined by the initial conditions $\left(\tilde{X}^\mu(\tau=0),\frac{d \tilde{X}^\mu(\tau=0)}{d\tau}|_{\tau =0}\right)$
and the external field. This is a simplified description, that do not consider back-reactions. However, we admit it just restricting the limit of applicability of the concept of test particle.

On the other hand, in Hamilton-Randers dynamics, the fundamental degrees of freedom interact during the $U_t$ dynamics. That there must be interactions between them is clearly manifest because the assumed structure of the cycles. Thus even in the case of an effective free quantum evolution, there are interactions among the sub-quantum degrees of freedom.
However, there is a natural notion of non-interacting quantum system in terms of the dynamics of the sub-quantum degrees of freedom: we say that a quantum system is not interacting with the environment if the interaction between the sub-quantum degrees of freedom of the system and the environment can be disregarded without changing the quantum and dynamical properties of the quantum system. In particular, the measure $\mu_P$ is preserved during the $U_t$ evolution and therefore, isomorphic on each cycle.

A free test particle is a test particle which is free in the above sense.

After the above propaedeutic discussion, we can state the following result:
\begin{proposicion}
Let $\mathcal{S}_i,\,i=1,2,3$ be three Hamilton-Randers dynamical systems with $N\gg 1$ associated sub-quantum degrees of freedom corresponding to free test particles. Assume that the relation \eqref{scales} holds good.
 Then the  macroscopic coordinates $\tilde{X}^\mu(\tau)$ do not depend upon the system $\mathcal{S}_i$.
 \label{proposiciononweakequivalenceprinciple}
\end{proposicion}
\begin{proof}
 The coordinate functions $\tilde{X}^\mu(\tau)$ define a world line in $M_4$ and are $1$-Lipschitz in the metastable domain $t\to (2n+1)T$ of each cycle. Then we can apply the concentration of measure \eqref{generalconcentrationofmeasure} on such domains, in each cycle. Hence the observable coordinates $\{\tilde{X}^\mu\}^4_{\mu=1}$ moves following the common $M^\mu(\tau)$ coordinates with an error bounded by $\exp(-C_2 N^2)$.
 \end{proof}
 Therefore, in the subset of metastable domain for $t= (2n+1)T$ there is a strong concentration for the value the functions $\{\tilde{X}^\mu(\tau)\}^4_{\mu=1}$  around the median $\{M^\mu(\tau)\}^4_{\mu=1}$. Note that this universality is up to fixing the initial conditions of the median $M^\mu$, which is equivalent to fix the initial conditions for $\{u^\mu_k\}^{N}_{k=1}$.
 Since $M^\mu(\tau)$ does not depend on the system $i=\,1,2,\mathcal{S}$, the above mathematical consequence of concentration, that is, the fact that the coordinate in free-fall evolution are guided by the median in such independent way of the details of the system, is interpreted as the weak equivalence principle or free fall.

Since the precision of the separation of the world lines given by the Chernoff's type bounds is non-zero, for individual sub-quantum degrees of freedom considered individually can be small violations of the weak equivalence principle. To understand the degree of such  possible violations, let us consider a system composed by a single, non-interacting sub-quantum degree of freedom. If one assumes that $N=4$ degrees of freedom, associated with four momentum and spin in the form of Dirac $4$-spinor wave function dynamical system, then the bound in the discrepancy in the evolution of local coordinates is of order of $\exp (- C 16)$.
Due to its emergent character, the theory predict the abrupt breakdown of the equivalence principle at scales where the degrees of freedom are associated with the fundamental degrees of freedom. This fix the constant $C$ to be of order $C\sim 1/16$. This is just an estimation of the {\it Chernoff's distance} $C$ based upon the assumptions of the theory and the expectation that there is no concentration at the fundamental scale.
On the other hand, for ordinary quantum systems, the theory indeed assumes $\gg 4$, then we have the prediction of almost exactness of the weak equivalence principle for any quantum or classical system, where the error in the exactness is a Gaussian error with $N^2$.

 The argument above is independent of the regularity properties of the macroscopic world line or of the nature of the source. In fact, we have not restricted to the case of a classical source. If the source is a quantum system, then our argument only states that any test particle will undertake an universal world line that only depends upon the initial conditions of the center of mass and its velocity. In the case of a quantum source, it is also reasonable that the test particle must be also a quantum system. Note then that our argument implies that the quantum particle will naturally collapse along the world line parameterized by $M^\mu(\tau)$, but there is no constraint on the regularity of $M^\mu(\tau)$, which does not exclude fluctuating world lines, if no any other constraint is imposed.

\section{Newtonian gravity is $1$-Lipschitz continuous}

In order to discuss conditions under which the gravitational interaction is $1$-Lipschitz, let us consider first the newtonian gravitational force between a mass point particle  with mass $m$ by a mass point particle with mass $M$ located at the origin of coordinates,
   \begin{align}
   F_N(\vec{x})=\,-G_N\,\frac{m\, M}{r^2}, \quad \vec{x}\in\,\mathbb{R}^3
   \label{Newtonlaw}
   \end{align}
  where $G_N$ is the Newton gravitational constant and $r=|\vec{x}|$ the distance to the origin in $\mathbb{R}^3$ of  the point $\vec{x}$.
   To compare different lengths or different mechanical forces, it is useful to consider dimensionless expressions, for which we need reference scales.
In doing this comparison we adopt as length scale the Planck length and for the force scale the Planck force and use homogenous quantities for length and force. The Planck force provides a natural unit, respect to which  we can compare any other scale.
The Lipschitz condition applied to the Newtonian force is of the form
\begin{align*}
\frac{|\vec{F}_N(\vec{r}_1)-\,\vec{F}_N(\vec{r}_2)|}{F_P}=\,\frac{G_N m M}{F_P}\Big|\frac{\vec{r}_1}{r^3_1}-\frac{\vec{r}_2}{r^3_2}\Big|<\frac{|\vec{r}_1-\vec{r}_2|}{l_P},
\end{align*}
where $F_P$ is the Planck force and $\l_P$ is the Planck length.
This leads, after some work, to a sufficient condition for $1$-Lipschitz force,
\begin{align*}
G_N\,m\,M\,\frac{l_P}{F_P}\,\frac{(r_1-r_2)}{(r_1+r_2)}\frac{1}{r_1 r_2}<1,\quad r_1>r_2.
\end{align*}
This condition is re-cast in a stronger from,
\begin{align}
2\,G_N\,m\,M\,\frac{\l_P}{F_P}\,\frac{1}{r^2}<1,
\end{align}
that leads to the numerical condition
\begin{align}
17.82\times 10^{-90} |m|\,|M|\,\frac{1}{|r|^2}<1,\quad r=\min\{r_1,r_2\}
\label{sufficiente condition for Lipschitz Newtonian gravity}
\end{align}
where $|m|$ and $|M|$ are measured in kg, while $|r|$ is measured in meters. $r$ is understood as the typical scale of the action of gravity in the system, but being rigourous, it runs in a maximal interval of the form $[\l_P,+\infty[$. All known physical systems satisfy the condition \eqref{sufficiente condition for Lipschitz Newtonian gravity} in regions far enough from the divergence at $r=0$, which is reasonable since $r$ indicates the scale of the system, instead than a variable.

\section{Consistency of modified gravity with emergent quantum mechanics}

As we have argued, the weak equivalence principle can be seen as a consequence of the postulates of Hamilton-Randers theory. Specifically, as a consequence of the postulates on the existence of a large complex dynamics at the sub-quantum level and the postulate on the existence of a domain where the fundamental dynamics is $1$-Lipschitz it was showed the compatibility of the Newtonian model of gravity with the $1$-Lipschitz condition of Hamilton-Randers theory. However, Newtonian gravity is not the only case of gravitational models with domains compatible in the above sense with Hamilton-Randers theory. In this section we discuss two extension of Newtonian gravity of the type of modified gravity from this point of view.

Let us consider here the case of a logarithmic type potential $\phi =\,\bar{k}_0\, m\,M\,\log r$. Then the corresponding  force is determined by a logarithmic function of the radial distance,
 \begin{align}
 \vec{F}_1(\vec{r})=\,-\bar{k}_0 m M \frac{\vec{r}}{r^2},
 \label{modified gravity law}
 \end{align}
 where $\bar{k}_0$ is a constant not depending upon $m$ and $M$.
 Then following an argument similar to the Newtonian case, a sufficient condition for $\vec{F}_1$ being a $1$-Lipschitz force can be cast in a strong form by the condition
\begin{align}
2\,\bar{k}_0\,m\,M\,\frac{\l_P}{F_P}\,\frac{1}{r}<1.
\end{align}

Newtonian gravity and modified Newtonian gravity are not the only models of non-relativistic gravity compatible with the condition of $1$-Lipschitz force. If one consider a convex combination of the forces for a $1/r$ type potential and a logarithmic potential, the corresponding force will be of the form
\begin{align*}
\vec{F}_\lambda=\, \lambda\vec{F}_2(\vec{r})+\,(1-\lambda)\vec{F}_1 (\vec{r}),\,\quad\, \lambda\in \,[0,1] .
\end{align*}
The third law of Newton dynamics is implemented if $F_2$ is proportional to the product $m\,M$. Thus the general form for $F_2$ is of the form $\vec{F}_2=\,G m M\,\lambda\,\frac{\vec{r}}{r^3}$.
The general combination of $\vec{F}_2$ type force and  $\vec{F}_1$ is of the form
\begin{align*}
\vec{F}_\lambda(\vec{r})=\,- G m M\,\lambda\,\frac{\vec{r}}{r^3}-\,(1-\lambda)\,\bar{k}_0\,m M \frac{\vec{r}}{r^2},\,\quad\, \lambda\in \,[0,1].
\end{align*}
The condition $\lambda \in \, [0,1]$ is necessary to assure that the force is $1$-Lipschitz; otherwise, if $\lambda$ is not constrained in magnitude, the domain of consistence with the $1$-Lipschitz condition will be small. Note that $G$ is a constant, but does not necessarily coincide with the Newton constant. Therefore, from a purely formal point of view, this expression for the force can be re-casted as
\begin{align}
\vec{F}(\vec{r})=\,- G_N m M\,\frac{\vec{r}}{r^3}-\,{k}_0\,m M \frac{\vec{r}}{r^2},
\label{Modified Newton force}
\end{align}
where $G_N=\lambda \,G$ is identified with the Newton Gravitational constant, while the constant ${k}_0=\,(1-\lambda)\bar{k}_0$ is independent of the masses $m$ and $M$.

One sufficient condition for the force $\vec{F}$ be consistent with the $1$-Lipschitz condition is the requirement that
\begin{align}
2\,G_N\,m\,M\,\frac{\l_P}{F_P}\,\frac{1}{r^2}+\,2\,k_0\,m\,M\,\frac{\l_P}{F_P}\frac{1}{r}<\,1.
 \label{sufficient condition for Lipschitz in modified gravity}
\end{align}
This expression is independent of $\lambda$ and depends only on observable quantities and observable constants and parameters.

The expression \eqref{Modified Newton force} for the modified gravitational force $\vec{F}$ implies that for large distance scales, the gravitational force is proportional to $1/r$, while at short distance (in astrophysical terms), the force is the usual Newtonian force. Therefore, there is a distance scale $r_c$ where the forces are approximately equal,
\begin{align}
r_c=\,\frac{G_N}{{k}_0}.
\end{align}
At such scale of distance, a probe system with mass $m$ testing a system with mass $M$ and characteristic radius $r_{ch}$ experiences an acceleration due to the force $F_1$ given by the relation
\begin{align}
a_{ch} (M,r_{ch})=\, M \frac{G_N}{r_c\,r_{ch}}.
\label{characteristic acceleration}
\end{align}
The {\it critical radius} $r_c$ is independent of the source mass $M$, but the characteristic acceleration  $a_{ch} (M,r_{ch})$ depends on $M$ and the characteristic scale $r_{ch}$.

One can provide estimates for the constant $k_0$ by using the same fit as in MOND. Let us consider the M33 Galaxy. If we assume that for this system the $a_{ch} (M,r_{ch}) $ is equal to MOND acceleration $a_{MOND}=\,1.2\times\,10^{-10}\,m/s^{-2}$ and if we assume that M33 has a mass of $M =\,1.5\times10^{41} kg$, then $r_{ch}=r_c$ is fixed by the relation $r_{c}=\, 7.933\times 10^5\,ly$. It also determines the value of the constant $k_0 =\,6.004\times \,10^{-20} \,m^2kg^{-1} s^{-2}.$ Adopting this value, the sufficient condition \eqref{sufficient condition for Lipschitz in modified gravity} reads numerically
\begin{align}
17.82\times 10^{-90} |m|\,|M|\,\frac{1}{|r|^2}+\,2.41\times 10^{-97}\,|m|\,|M|\frac{1}{|r|}<\,1.
\label{numerical Lipschitz condition in modified gravity}
\end{align}
This condition is fulfilled at all known physical scales in regions far enough from the divergence at $r=0$.

Note that $r_c$ is of the order of the typical spiral galaxies diameter. This observation is a direct consequence of the force law \eqref{Modified Newton force}. Besides the quality of the experimental data could not be the desired to reach a direct conclusion, Fig. 48 in ref. \cite{Famaey and McGaugh} indicates a possible fail of the pure MOND starting just under the scale of the spiral galaxies.
The critical radius $r_c$ yields to a critical acceleration for a source of mass $M$ of the form
\begin{align}
a_{ch} (M,r_{ch})=\,6.3\times\,10^{-51}\,|M|\,m/s^2,
\label{critical acceleration}
\end{align}
where $|M|$ is measured in $kg$.
However, in the regime where $|\vec{F}_1|>>|\vec{F}_2|$, equating the force to the centripetal force, we have
\begin{align*}
k_0\,m\,  M\, \frac{1}{r_{ch}}=\,m\,\frac{\mathrm{v}^2}{r_{ch}},
\end{align*}
or
\begin{align*}
\mathrm{v}^4=\,k^2_0 \,M^2 .
\end{align*}
By means of the characteristic acceleration \eqref{characteristic acceleration}, this expression is re-casted as
\begin{align}
\mathrm{v}^4 =\, M G_N\,a_{ch}(M,r_{ch})\, \frac{r_{ch}}{r_c}.
\label{modified asymptotic Tully Fisher}
\end{align}
Therefore, the ratio $\frac{v^4}{M\,G_N}$ is given by the expression
\begin{align*}
\frac{v^4}{M\,G_N} =\,a_{ch}(M,r_{ch})\, \frac{r_{ch}}{r_c}=\,\frac{G_N}{r^2_c}\,M.
\end{align*}
Comparing with experimental data (see Fig. 48 in \cite{Famaey and McGaugh}), the linearity on the mass $M$ in the second side of this relation disagrees with current experimental observations.

The second class of modified gravity models that we  consider and that are compatible with the $1$-Lipschitz condition are such that the potential dominating at large  distance scales is of the form $\Phi=\,-\bar{\kappa}\,m\,\sqrt{M}\,\log r$. Then the gravitational force law  is
\begin{align}
\vec{F}'_1=\, -\bar{\kappa}\,m\,\sqrt{M}\,\frac{\vec{r}}{r^2}.
\label{alternative law of modified newton dynamics}
\end{align}
This force law coincides with deep MOND.
For large scales, this form of the force is manifestly weaker than the force $\vec{F}_1$ discussed before. Therefore, it is also consistent with the $1$-Lipschitz condition.

To shorter scales we assume the Newtonian law of force. Thus the second ansatz for the force is
\begin{align}
\vec{F}(\vec{r})=\,- G_N m M\,\frac{\vec{r}}{r^3}-\,\bar{\kappa}\,m \sqrt{M} \frac{\vec{r}}{r^2},
\label{Modified Newton force2}
\end{align}
By an analysis similar to the previous case, the critical radius where $|\vec{F}'_1 |=\,|\vec{F}_2 |$ leads to the relation
\begin{align}
r_c (M) =\,\frac{G_N}{\bar{\kappa}}\,\sqrt{M}.
\end{align}
In this case, there is no universal (independent of $M$) critical radius. In contrast, the associated critical acceleration is
\begin{align}
a_c= \,\bar{\kappa}\,m\,  \sqrt{M}\, \frac{1}{r_c} =\,\frac{\bar{\kappa}^2}{G_N}.
\end{align}
This critical acceleration is independent of $M$ in the ambit of application of the force law \eqref{Modified Newton force2}.

Note that the force law \eqref{alternative law of modified newton dynamics} is inconsistent with the third law of Newtonian dynamics. Indeed the condition for the third law leads to the constraint
\begin{align*}
- G_N m M\,\frac{\vec{r}}{r^3}-\,\bar{\kappa}\,m \sqrt{M} \frac{\vec{r}}{r^2}=\,- G_N m M\,\frac{\vec{r}}{r^3}-\,\bar{\kappa}\,\sqrt{m} {M} \frac{\vec{r}}{r^2},
\end{align*}
 which is true for all values of $r$ if and only if $M=m$. Note also that the scale where the violation of the third law starts being manifestly large is when one of the forces $\vec{F}'_1$ dominates. For $m< M$ this happens in a large amount from the critical scale $r_c(m)=\frac{G_N}{\bar{\kappa}}\,\sqrt{m}< r_c(M)$.

 For the rotating galaxy, under the condition that the force $\vec{F}'_1 $  dominates, the force must coincide with the centripetal force,
\begin{align*}
\bar{\kappa}\,m\,  \sqrt{M}\, \frac{1}{r_{ch}}=\,m\,\frac{\mathrm{v}^2}{r_{ch}}.
\end{align*}
This leads to the relation
\begin{align}
\frac{v^4}{M} =\,\bar{\kappa}^2 .
\end{align}
When this relation is interpreted as the barionic Tully-Fisher relation, then the value of the acceleration is fixed by observed dynamics of spiral galaxies. The best fit is for the value $a_c\approx 1.2\,\times 10^{-10} m/s^2 =\, a_{MOND}$ and the value of $\bar{\kappa}^2$ is fixed to be $\bar{\kappa}^2=\,G_N \,a_{MOND}\simeq \,8\times 10^{-21}\,m^4 kg^{-1} s^{-4}$.

Let us remark that the two models discussed above are not all the possible models compatible with the $1$-Lipschitz condition of Hamilton-Randers dynamics. Indeed, convex combinations of compatible models are also compatible. On the other hand, the possible laws cannot diverge too large when $r\to 0$, otherwise will not be $1$-Lipschitz with enough generality, that is, for all possible systems. Thus consistency with the $1$-Lipschitz condition at short scales limits the type of models. Therefore, the most general models of force law compatible with the $1$-Lipschitz principle of Hamilton-Randers dynamics is of the form
\begin{align}
\vec{F}=\,\int_{Mod} \rho_{M}(\varsigma)\,\vec{F}_{\varsigma},
\label{general form of compatible force}
\end{align}
where formally $Mod$ is the moduli space of individual vector forces compatible with the $1$-Lipschitz conditions and $\rho_{M}\in [0,1[$ is a density on $Mod$ such that $\int_{Mod} \,\rho_M =1$. The force law of Newton \eqref{Newtonlaw} and the modified gravity laws \eqref{Modified Newton force} and \eqref{Modified Newton force2} are particular cases.
\section{Discussion and Conclusion}
We have shown that, as a result of concentration of measure in Hamilton-Randers dynamical systems, the weak equivalence principle is understood as an emergent law of Nature. According to our reasoning, the weak equivalence principle must be an exact law, valid for systems at any quantum or classical scales. Therefore, one should not expect any observable violation of the principle with a very high precision until the system under study reach the fundamental scale describing the dynamics of individual sub-quantum systems. At such scale, the weak equivalence principle is totally violated, since the principle of concentration of measure due to the large dimensionality of the configuration space associated to a fundamental degree of freedom does not apply. This precise way on how the weak equivalence principle holds is in principle a falsifiable prediction of our theory.

A main insight explored in this work is the possibility to use the $1$-Lipschitz condition of the fundamental dynamics as a guiding principle in the search of gravity models compatible with Hamilton-Randers theory. In this context, we have argued that Newtonian gravity is a $1$-Lipschitz interaction. This is also the case for certain models of modified gravity. We have discussed two models of modified gravity. In the first one, a direct generalization of Newton theory for the case of a logarithmic potential field is discussed. For this first model, while the critical radius $r_c$ is constant, the characteristic acceleration $a_{ch} (M,r_{ch})$ depends upon the mass of the  source $M$. The model is fully consistent with the third's law of Newtonian dynamics. However, fixing the value of $r_c$ by means of fitting with respect to the  size of the well known $M33$ galaxy, we have seen that the modified gravity law \eqref{modified gravity law} is inconsistent with the experimental barionic Tully-Fisher relation.

The second model discussed in the paper, based upon the law \eqref{alternative law of modified newton dynamics}, is partially consistence with the empirical barionic Tully-Fisher relation, by fixing the constant $\bar{\kappa}$ with the observed velocity curve for spiral galaxies. A main discrepancy between theory and experiment (as in the case of MOND) is related with the apparent notorious {\it jump} in the observed barionic Tully-Fisher relation in the transition between spiral galaxies and galaxy groups. This jump, which of order of a factor $2$, remains to be explained, either by adding a different force at such scales compatible with the $1$-Lipschitz condition or, as suggested by Milgrom, by applying known physics to processes still to be understood better at such scales \cite{Milgrom2014}. It is also remarkable that the model explicitly violates the third law of Newtonian dynamics. This implies a violation of the conservation for linear momentum for systems of masses $m$ and $M$ interacting gravitationally at any scale, although it is remarkable at scales where the modified gravitation law is dominant with respect to the Newtonian law.

This second model features several common elements with MOND: first, the modification of Newton's gravitation force law is described by means of a logarithmic potential, second, the model is consistent with weak equivalence principle at the scales where the force \eqref{Modified Newton force2} is applicable, third, at first sight, the model implies a violation of momentum conservation. Nevertheless, the model discussed here and MOND are formally different. Indeed, we formulate the model by means of the force law \eqref{Modified Newton force2}, which quite different from MOND's force law \cite{Milgrom 1983a}. Thus despite the similarities, the discrepancies in the structure and formal law of force should refrain the fully identification of our second model with MOND theory.

We have also shown that only imposing the $1$-Lipschitz constraint does not fix the force law. This is apparent by the general form of the possible modified gravitational models as descried by the relation \eqref{general form of compatible force} and the particular examples of modified gravity described in this paper.

In conclusion, we have shown that Hamilton-Randers theory implies the weak equivalence principles and that it accommodates several forms of gravitational laws as emergent interactions. Significatively, Newton's gravitational force law and certain modified gravity force laws are consistent with the principles of Hamilton-Randers theory. One of the models of gravity discussed has remarkable common features with MOND, but it is still unclear up to which extends the model discussed can be identified with MOND. This question and others regarding the methods described in this paper are interested problems for further research. Specifically, it is interesting to  understand how to complete the constraints to determine completely the models compatible with Hamilton-Randers and how to extend the models discussed in this paper, for instance, by considering the models of the type described in \eqref{general form of compatible force}, to provide an accurate and natural description of the barionic Tully-Fisher relation over an ample domain including galaxy groups and galaxy clusters.


\begin{thebibliography}{22}

\bibitem{Adler}  S. L. Adler, {\it Quantum Theory as an Emergent Phenomenon: The Statistical Mechanics of Matrix Models as the Precursor of Quantum Field Theory}, Cambridge University Press (2004).

\bibitem{Blasone2}M. Blasone, P. Jizba and F. Scardigli, {\it Can Quantum Mechanics be an Emergent
Phenomenon?}, J. Phys. Conf. Ser. {\bf 174} (2009) 012034, arXiv:0901.3907[quant-ph].

\bibitem{Elze} H.T. Elze {\it Quantum mechanics emerging from "timeless" classical dynamics},
 Trends in General Relativity and Quantum Cosmology, ed. C.V. Benton (Nova Science Publ., Hauppauge, NY, 79-101 (2006); H.T. Elze,
{\it The Attractor and the Quantum States}, Int. J. Qu. Info. {\bf 7}, 83 (2009);  H.T. Elze,
{\it Symmetry Aspects in Emergent Quantum Mechanics}, J. Phys. Conf. Ser. {\bf 171}, 012034 (2009).

\bibitem{Ricardo2005} R. Gallego Torrom\'e, {\it Quantum Systems as results of Geometric Evolutions},  arXiv:math-ph/0506038v6 .

 \bibitem{Ricardo2006}R. Gallego Torrom\'e, {\it A Finslerian version of 't Hooft Deterministic Quantum Models},
J. Math. Phys. {\bf  47}, 072101 (2006).

\bibitem{Ricardo2014} R. Gallego Torrom\'e, {\it Foundations for a theory of emergent quantum mechanics and emergent classical gravity}, Last version available at https://www.researchgate.net/publication/286923485; Previous versions available at arXiv:1402.5070 [math-ph].

\bibitem{Groessing2013} G. Gr\"{o}ssing, {\it Emergence of Quantum Mechanics from a Sub-Quantum Statistical Mechanics}, Int. J. Mod. Phys. B, {\bf 28}, 1450179 (2014).

\bibitem{Hooft} G. 't Hooft, {\it Determinism and Dissipation in Quantum Gravity}, Class. Quantum Grav. {\bf 16}, 3263 (1999).

\bibitem{Hooft 2001} Gerard't Hooft, {\it How does God play dies?(Pre-) Determinism at the Planck Scale}, {
hep-th/0104219}.

\bibitem{Hooft2006} G. 't Hooft, {\it Emergent Quantum Mechanics and Emergent Symmetries}, 13th International Symposium on Particles, Strings, and Cosmology-PASCOS 2007. AIP Conference Proceedings {\bf 957}, pp. 154-163 (2007).

\bibitem{Hooft2012wavefunctionschroedingercollapse} G. 't Hooft, {\it How a wave function can collapse without violating Schr\"odinger's equation, and how to understand the Born's rule}, arXiv:1112.1811v3.

\bibitem{Hooft2016} G. 't Hooft, {\it The Cellular Automaton Interpretation of Quantum Mechanics}, Fundamental Theories in Physics Vol. {\bf 185}, Springer Verlag (2016).

\bibitem{Smolin2012} L. Smolin, {\it A real ensemble interpretation of quantum mechanics}, Found. of Phys. {\bf 42}, Issue 10, pp 1239–1261 (2012).

\bibitem{Koopman1931} B. O. Koopman, {\it Hamiltonian Systems and Transformations in Hilbert Space} Proceedings of the National Academy of Sciences {\bf 17} (5), 315 (1931).

\bibitem{von Neumann} J. von Neumann,  {\it Zur Operatorenmethode In Der Klassischen Mechanik}, Annals of Mathematics {\bf 33} (3): 587–642 (1932);
J. von Neumann,  {\it Zusatze Zur Arbeit "Zur Operatorenmethode.... }, Annals of Mathematics {\bf 33} (4): 789–791 (1932).

\bibitem{Ricardo2019a} R. Gallego Torrom\'e, {\it On the origin of the weak equivalence principle in a theory of emergent quantum mechanics}, Int. J. Geom. Methods Mod. Phys., Vol. {\bf 17}, No. 10, 2050157 (2020).

\bibitem{Bell1964} J. S. Bell, {\it On the Einstein-Podolski-Rosen paradox}, Physica {\bf 1}, 195 (1964).

\bibitem{Isham 1995} C. Isham, {\it Lectures on quantum theory: mathematical and theoretical foundations}, Imperial College Press (1995).

\bibitem{Redhead} M. Redhead, {\it Incompleteness, Nonlocality and Realism: A prolegomenon to the Philosophy of Quantum Mechanics},  Claredon Press, Oxford (1989).

\bibitem{Gromov} M. Gromov, {\it Riemannian structures for Riemannian and non-Riemannian spaces}, Birkh$\ddot{a}$user (1999).

\bibitem{MilmanSchechtman2001} V. D. Milman and G. Schechtman, {\it Asymptotic theory of Finite Dimensional normed spaces}, Lecture notes in Mathematics 1200, Springer (2001).

\bibitem{Talagrand} M. Talagrand, {\it A new look at independence}, Ann. Probab. {\bf 24}, Number 1, 1-34 (1996).

\bibitem{Berger 2002} M. Berger {A panoramic view of Riemannian Geometry}, Springer-Verlag (2002).

\bibitem{Milgrom 1983a} M. Milgrom, {\it A modification of the Newtonian dynamics as a possible alternative to the hidden mass hypothesis}, Astrophysical Journal., {\bf 270}, 365 (1983).
\bibitem{Milgrom 1983b} M. Milgrom, {\it A modification of the Newtonian dynamics - Implications for galaxies}, Astrophysical Journal., {\bf 270}, 371 (1983).

\bibitem{Milgrom 2015} M. Milgrom, {\it MOND Theory}, Canadian Journal of Physics, {\bf 93(2)}: 107 (2015).

\bibitem{Bekenstein Milgrom} J. Bekenstein and M. Milgrom, {\it Does the missing mass problem signal the breakdown of Newtonian gravity?}, Astrophys. J., {\bf 286}, 7 (1984).

\bibitem{Singh} T. P. Singh, {\it From quantum foundations, to spontaneous quantum gravity: an overview of the new theory}, Z. Naturforsch. A {\bf 75}, 833 (2020).

\bibitem{Ricardo2017b} R. Gallego Torrom\'e, {\it Emergent Quantum Mechanics and the Origin of Quantum Non-local Correlations},
International Journal of Theoretical Physics, {\bf 56}, 3323(2017).

\bibitem{Arnold} V. Arnold, {\it Mathematical Methods of Classical Mechanics}, Springer (1989).

 \bibitem{Bars2001} I. Bars, {\it Survey of Two-Time Physics}, Class.Quant.Grav. {\bf 18} 3113-3130 (2001).

 \bibitem{Ricardo2023} R. Gallego Torrom\'e, J.M. Isidro, Pedro Fern\'andez de C\'ordoba, {\it On the emergent origin of the inertial mass}, Foundations of Physics,  {\bf 53}, Article number: 52 (2023).

 \bibitem{Randers} G. Randers, {\it On an Asymmetrical Metric in the
Four-Space of General Relativity}, Phys. Rev. {\bf 59},
195-199 (1941).

\bibitem{Dolce} D. Dolce, {\it Elementary Spacetime cycles}, Europhys. Lett., {\bf 102}, 31002 (2013); {\it Internal times” and how to second-quantize fields by means of periodic boundary conditions}, Annals of Physics
Volume {\bf 457}, October 2023, 169398 .

\bibitem{Famaey and McGaugh} Benoit Famaey and Stacy McGaugh, {\it Modified Newtonian Dynamics (MOND): Observational Phenomenology and Relativistic Extensions}, Living Reviews in Relativity, {\bf 15}, 10 (2012).

\bibitem{Milgrom2014} M. Milgrom, {\it The MOND paradigm of modified dynamics},  Scholarpedia, 9(6):31410 (2014).







\end{thebibliography}
\end{document}